\documentclass[twocolumn,pra,floatfix,amsmath,superscriptaddress]{revtex4-1}
\usepackage[latin9]{inputenc}
\usepackage{graphicx}
\usepackage[hypertex]{hyperref}
\usepackage{bm,bbm}
\usepackage{amsmath}
\usepackage{amsfonts}
\usepackage{amsthm}
\usepackage{amssymb}
\usepackage{mathrsfs}
\usepackage{subfigure}
\usepackage{array}
\usepackage{color}
\makeatletter

\usepackage{pxfonts}

\newtheorem{theorem}{Theorem}

\begin{document}

\title{Quantum coherence fraction}

\author{Yao Yao}
\email{yaoyao@mtrc.ac.cn}
\affiliation{Microsystems and Terahertz Research Center, China Academy of Engineering Physics, Chengdu Sichuan 610200, China}
\affiliation{Institute of Electronic Engineering, China Academy of Engineering Physics, Mianyang Sichuan 621999, China}

\author{Dong Li}
\affiliation{Microsystems and Terahertz Research Center, China Academy of Engineering Physics, Chengdu Sichuan 610200, China}
\affiliation{Institute of Electronic Engineering, China Academy of Engineering Physics, Mianyang Sichuan 621999, China}

\author{C. P. Sun}
\affiliation{Graduate School of China Academy of Engineering Physics, Beijing 100193, China}
\affiliation{Beijing Computational Science Research Center, Beijing 100193, China}

\date{\today}

\begin{abstract}
As an analogy of \textit{fully entangled fraction} in the framework of entanglement theory, we have introduced the notion of
\textit{quantum coherence fraction} $C_{\mathcal{F}}$, which quantifies the closeness between a given state and the set of maximally coherent states. By providing an
alternative formulation of the robustness of coherence $C_{\mathcal{R}}$, we have elucidated the relationship between quantum coherence fraction and the normalized version of
$C_{\mathcal{R}}$ (i.e., $\overline{C}_{\mathcal{R}}$), where the role of genuinely incoherent operations (GIO) is highlighted. Numerical simulation shows that though as expected $C_{\mathcal{F}}$
is upper bounded by $\overline{C}_{\mathcal{R}}$, $C_{\mathcal{F}}$ constitutes a good approximation to $\overline{C}_{\mathcal{R}}$ especially in low-dimensional Hilbert spaces.
Even more intriguingly, we can analytically prove that $C_{\mathcal{F}}$ is exactly equivalent to $\overline{C}_{\mathcal{R}}$ for \textit{qubit} and \textit{qutrit} states.
Moreover, some intuitive properties and implications of $C_{\mathcal{F}}$ are also indicated.
\end{abstract}

\maketitle
\section{INTRODUCTION}
The concept of quantum coherence has played a prominent role in the development of quantum physics
and can be viewed as an essential resource for almost all applications in quantum information processing \cite{Streltsov2017,Hu2018}.
Quite recently, the characterization and quantification of quantum coherence has become one of the most active areas of
current research in the field of quantum information theory
\cite{Baumgratz2014,Yao2015,Streltsov2015,Winter2016,Napoli2016,Piani2016,Chitambar2016a,Chitambar2016b,Yadin2016,Marvian2016,Vicente2017}.
The recent significant breakthrough lies in the successful application of quantum resource theory in this endeavor,
elevating the story of quantum coherence to a quantitative theory in a mathematically rigorous framework \cite{Streltsov2017,Chitambar2019}.

However, compared to its celebrated predecessor, the resource theory of quantum entanglement \cite{Vedral1997,Vedral1998,Plenio2007,Horodecki2009},
two subtle but crucial differences emerge: (i) in the framework proposed in Ref. \cite{Baumgratz2014}, the exact value of quantum coherence is
defined with respect to a \textit{prefixed} basis, while all the entanglement measures are invariant under local unitaries, that is,
a local change of basis leave quantum entanglement unchanged \cite{Vedral1997,Vedral1998}; (ii) in the theory of quantum entanglement,
the free states (e.g., separable states) and the free operations (i.e., local operations and classical communication, or LOCC for short) are \textit{naturally} specified
by the LOCC constraint, which is both technologically and fundamentally well-motivated \cite{Plenio2007,Horodecki2009},
while except for the incoherent operations (IO) \cite{Baumgratz2014}, alternative proposals of free operations have also been put forward to impose various constraints on
the resource theory of quantum coherence, such as the maximal incoherent operations (MIO) \cite{Aberg2006a,Aberg2006b},
dephasing-covariant incoherent operations (DIO) \cite{Marvian2016,Chitambar2016a},
strictly incoherent operations (SIO) \cite{Winter2016,Yadin2016} and genuinely incoherent operations (GIO) \cite{Vicente2017},
which are meaningful in their own right but somehow artificial and physically inconsistent \cite{Chitambar2016a,Chitambar2016b}
(see Fig. \ref{fig1} for clarity).

\begin{figure}[htbp]
\begin{center}
\includegraphics[width=0.30\textwidth]{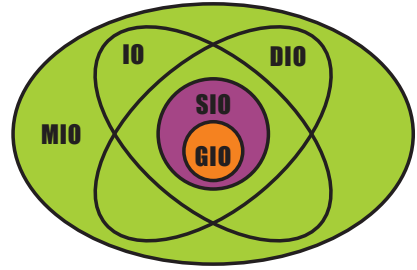}
\end{center}
\caption{(Color online) Hierarchical relationships between MIO, DIO, IO, SIO and GIO \cite{Chitambar2016a,Chitambar2016b,Yao2017}.
}\label{fig1}
\end{figure}

On the other hand, nearly all of the coherence measures have been established within the above framework, including
the $l_1$ norm of coherence, the relative entropy of coherence and the robustness of coherence \cite{Baumgratz2014,Napoli2016},
and meanwhile we become explicitly aware of the fact that the coherence measures are closely related to their analogs  in entanglement theory
\cite{Streltsov2015,Zhu2017}. Nevertheless, it is worth noting that in some scenarios it may \textit{not} be necessary
to directly evaluate entanglement or coherence measures of a state in order to quantify the quantum information processing \textit{usefulness} of a state.
In this sense, one quantifier that attracts our attention is the \textit{fully entangled fraction} (FEF) (or called \textit{singlet fraction})
\cite{Bennett1996}:
\begin{align}
\mathcal{F}(\rho)&=\max_{|\Phi\rangle\in\mathcal{MES}}\langle\Phi|\rho|\Phi\rangle,\label{FEF}\\
&=\max_{U,V}\langle\Phi^+|(U^\dagger\otimes V^\dagger)\rho(U\otimes V)|\Phi^+\rangle,
\end{align}
where $\rho\in\mathcal{H}(\mathbb{C}^d\otimes\mathbb{C}^d)$, $U(V)$ is a local unitary operation, $\mathcal{MES}$ is the set of maximally entangled
states and $|\Phi^+\rangle=1/\sqrt{d}\sum_{i=0}^{d-1}|ii\rangle$.

Notice that in general the FEF is not an entanglement monotone since
LOCC can enhance the value of FEF \cite{Badziag2000}. However, in fact the FEF is not only intimately related to entanglement distillation \cite{Horodecki1997,Horodecki1999a}
and teleportation \cite{Horodecki1999b,Albeverio2002}, but also to many other significant quantum information protocols such as dense coding, entanglement swapping
and Bell inequalities \cite{Grondalski2002}. More recently, this quantity has been identified as a useful measure in characterizing different
nonlocal correlations \cite{Cavalcanti2013,Hsieh2016,Quintino2016} and even involved in work extraction in quantum thermodynamics \cite{Hsieh2017}
and device-independent state estimation \cite{Bardyn2009}. Therefore, it would be intuitively motivated to introduce a counterpart
in coherence theory, that is, \textit{quantum coherence fraction}, and pursue its properties and implications.

This paper is organized as follows.
In Sec. \ref{sec2}, we explicitly present the definition of quantum coherence fraction $C_{\mathcal{F}}$ and list some of its properties,
illustrating the relations between $C_{\mathcal{F}}$ and other important coherence measures.
In Sec. \ref{sec3}, we offer an alternative formulation of the robustness of coherence $C_{\mathcal{R}}$, and consequently,
we discuss the relationship between $C_{\mathcal{F}}$ and the normalized version of $C_{\mathcal{R}}$,
where an inequality is established and the role of GIO is highlighted.
In Sec. \ref{sec4}, we provide a detailed numerical analysis of this inequality, and intriguingly,
we can analytically prove that $C_{\mathcal{F}}$ is exactly equivalent to $\overline{C}_{\mathcal{R}}$ for qubit and qutrit states.
Discussions and final remarks are given in Sec. \ref{sec5} and several open questions are raised for future research.

\section{Definition and basic properties}\label{sec2}
Throughout this paper, we consider the $\textit{d}$-dimensional Hilbert space $\mathcal{H}(\mathbb{C}^d)$,
let $\mathcal{D}(\mathbb{C}^d)$ be the convex set of density operators acting on $\mathcal{H}(\mathbb{C}^d)$,
and adopt the computational basis $\{|i\rangle\}_{i=0}^{d-1}$ as the incoherent basis \cite{Baumgratz2014}.
Thus, all diagonal density operators in this basis constitute the set of incoherent states
\begin{align}
\mathcal{I}=\left\{\rho\in\mathcal{D}(\mathbb{C}^d) \, | \, \rho=\Delta(\rho)\right\},
\end{align}
where $\Delta(\rho)=\sum_i|i\rangle\langle i|\rho|i\rangle\langle i|$ represents the completely decohering channel.
Any incoherent operation admits a Kraus representation $\Phi(\rho)=\sum_iK_i\rho K_i^\dagger$ such that every Kraus
operator is required to fulfill $K_i\rho K_i^\dagger/\textrm{tr}(K_i\rho K_i^\dagger)\in\mathcal{I}$ for all $\rho\in\mathcal{I}$  \cite{Baumgratz2014}.
For other important types of incoherent operations, we refer the readers to the review article \cite{Streltsov2017} for
their properties and relations among each other.
For later discussion, we recall two commonly used coherence measures, that is, the $l_1$ norm of coherence
and the robustness of coherence \cite{Baumgratz2014,Napoli2016}:
\begin{gather}
C_{l_1}(\rho)=\sum_{i\neq j}|\rho_{ij}|=\sum_{i,j}|\rho_{ij}|-1,\\
C_{\mathcal{R}}(\rho)= \min_{\tau \in {\mathcal{D}}(\mathbb{C}^d)} \left\{ s\geq 0\ \Big\vert\ \frac{\rho + s\ \tau}{1+s} =: \delta \in \mathcal{I}\right\}.
\end{gather}

In comparison with the definition of the FEF, we propose an analog quantity called \textit{quantum coherence fraction} (QCF):
\begin{align}
C_{\mathcal{F}}(\rho)=\max_{|\phi\rangle\in\mathcal{MCS}}\langle\phi|\rho|\phi\rangle
=\max_{U\in\mathcal{U}_d}\langle\phi^+|U^\dagger\rho U|\phi^+\rangle,
\label{QCF}
\end{align}
where the optimization is over all maximally coherent states ($\mathcal{MCS}$), $|\phi^+\rangle=1/\sqrt{d}\sum_{i=0}^{d-1}|i\rangle$ and
$\mathcal{U}_d$ is the set of diagonal unitary operators. Note that all maximally coherent states are of the form $|\phi\rangle=U_{inc}|\phi^+\rangle$,
where $U_{inc}=\sum_{i=0}^{d-1}e^{\textrm{i}\theta_i}|\pi(i)\rangle\langle i|$ and $\{\pi(i)\}$ is a permutation of $\{i\}$ \cite{Peng2016}.
Here the optimization over all incoherent unitary operators is not necessary, since $|\phi^+\rangle$ is invariant under all permutations.
Intuitively, $C_{\mathcal{F}}(\rho)$ measures how close a given state $\rho$ is to any maximally coherent state and remains unchanged
under diagonal incoherent unitary operations.

In the following, we present some elementary properties of the QCF and discuss its relationship with other coherence quantifiers
(see Appendix \ref{A1} for the proof):

(i) (\textit{Convexity}). $C_{\mathcal{F}}(\rho)$ is convex in $\rho$, that is, for any convex decomposition of a density operator $\rho=\sum_ip_i\rho_i$,
\begin{align}
C_{\mathcal{F}}(\rho)\leq\sum_ip_iC_{\mathcal{F}}(\rho_i).
\end{align}

(ii) For any $\rho\in\mathcal{D}(\mathbb{C}^d)$, we have
\begin{align}
\frac{1}{d}\leq C_{\mathcal{F}}(\rho)\leq\lambda_{max}\leq 1,
\end{align}
where $\lambda_{max}$ is the largest eigenvalue of $\rho$, $C_{\mathcal{F}}(\rho)=1$ if and only if $\rho$ is a maximally coherent state,
and if $\rho$ is an incoherent state, then $C_{\mathcal{F}}(\rho)=1/d$ (see Appendix \ref{A1} for more details).

(iii) For any $\rho\in\mathcal{D}(\mathbb{C}^d)$, $C_{\mathcal{F}}(\rho)$ is upper bounded by the coherence number of $\rho$ in the sense that
\begin{align}
C_{\mathcal{F}}(\rho)\leq\frac{N_c(\rho)}{d},
\end{align}
where the coherence number is defined as $N_c(\rho)=\min_{\{p_i,|\psi_i\rangle\}}\max_iR_c(|\psi_i\rangle)$,
$R_c(|\psi\rangle)$ denotes the coherence rank for pure states, and the optimization in $N_c(\rho)$ is over all pure-state convex decompositions of $\rho$
\cite{Killoran2016,Chin2017,Regula2018}.

(iv) For any $\rho\in\mathcal{D}(\mathbb{C}^d)$, $C_{\mathcal{F}}(\rho)$ is related to a SIO monotone $\mu_d(\rho)$ by the inequality
\begin{align}
C_{\mathcal{F}}(\rho)\leq\frac{1}{d}2^{\mu_d(\rho)},
\end{align}
where $\mu_d(\rho)=D_{max}(\rho\|\Delta(\rho))$ and
$D_{max}(\rho\|\sigma)$ denotes the quantum max-relative entropy between $\rho$ and $\sigma$ \cite{Datta2009}.
Note that $\mu_d(\rho)$ has played a key role in investigating the SIO distillable coherence \cite{Lami2019}.

Beside the above facts, we can evaluate $C_{\mathcal{F}}$ for some particular classes of states. For any pure state $|\psi\rangle=\sum_{i=0}^{d-1}c_i|i\rangle$,
one can obtain
\begin{align}
C_{\mathcal{F}}(|\psi\rangle)=\frac{1}{d}\left(\sum_{i=0}^{d-1}|c_i|\right)^2,
\label{pure}
\end{align}
where in the definition of Eq. (\ref{QCF}) we can choose $U_d=\textrm{diag}\{e^{\textrm{i}\arg(c_i)}\}$.
A general qubit state can be parameterized as
\begin{align}
\rho=
\left(\begin{array}{cc}
p & r \\
r^\ast & 1-p
\end{array}\right)
\end{align}
where $0\leq p\leq 1$ and the off-diagonal entry $r=|r|e^{\textrm{i}\arg(r)}$ with $|r|\leq\sqrt{p(1-p)}$. It is easy to show that
the optimal diagonal unitary operator is $U_d=\textrm{diag}\{e^{\textrm{i}\arg(r)/2},e^{-\textrm{i}\arg(r)/2}\}$, which
transforms $\rho$ into a \textit{real} positive matrix with the magnitude $|r|$ as the off-diagonal entry.
Therefore, for qubit states we have
\begin{align}
C_{\mathcal{F}}(\rho)=\frac{1+C_{l_1}(\rho)}{2}=\frac{1+C_{\mathcal{R}}(\rho)}{2}=\frac{1+2|r|}{2}.
\end{align}
Therefore, if we define the normalized versions of $C_{l_1}(\rho)$ and $C_\mathcal{R}(\rho)$
\begin{align}
\overline{C}_{l_1}(\rho)=\frac{1+C_{l_1}(\rho)}{d},\, \overline{C}_\mathcal{R}(\rho)=\frac{1+C_{\mathcal{R}}(\rho)}{d},
\end{align}
the above results suggest that for any pure or qubit states we have $C_{\mathcal{F}}(\rho)=\overline{C}_{\mathcal{R}}(\rho)=\overline{C}_{l_1}(\rho)$.
Actually, we can prove the following theorem.

\begin{theorem}
Let $\rho \in {\mathcal{D}}(\mathbb{C}^d)$ be a state such that there exists a unitary operator $U \in \mathcal{U}_d$,
which maps $\rho$ into $\rho' = U^{\dagger}\rho U$ with entries $\rho'_{ij} = |\rho_{ij}|$. Then
\begin{align}
C_{\mathcal{F}}(\rho)=\overline{C}_{\mathcal{R}}(\rho)=\overline{C}_{l_1}(\rho).
\end{align}
\label{T1}
\end{theorem}

\begin{proof}
In fact, if there exists a unitary operator $U \in \mathcal{U}_d$,
which maps $\rho$ into $\rho' = U^{\dagger}\rho U$ with entries $\rho'_{ij} = |\rho_{ij}|$, the following chain of (in)equalities are established
\begin{align}
C_{\mathcal{F}}(\rho)&\geq\langle\phi^+|U^\dagger\rho U|\phi^+\rangle=\textrm{tr}(|\phi^+\rangle\langle\phi^+|\rho')\nonumber\\
&=\frac{1}{d}\sum_{i,j}\rho'_{ij}=\frac{1}{d}\sum_{i,j}|\rho_{ij}|=\overline{C}_{l_1}(\rho).\nonumber
\end{align}
On the other hand, using the dual form of the semidefinite program (SDP) representation of $C_{\mathcal{R}}(\rho)$ \cite{Piani2016,Chitambar2016b}
\begin{align}
1+C_{\mathcal{R}}(\rho)&=\min_{\sigma\in\mathcal{I}}\{s \,|\, \rho\leq s\sigma \}\\
&=\min_\sigma\{\textrm{tr}(\sigma) \,|\, \rho\leq\sigma, \sigma=\Delta(\sigma)\}\\
&=\max_\tau\{\textrm{tr}(\rho\tau)\,|\, \tau\geq0, \Delta(\tau)=\openone\}\\
&\geq\textrm{tr}(\rho \mathcal{J}),
\end{align}
where $\mathcal{J}=d|\phi\rangle\langle\phi|$ and $|\phi\rangle$ is an arbitrary $\mathcal{MCS}$. Since
$C_{\mathcal{F}}(\rho)$ is the maximum value of the right-hand side of the inequality, we finally get $C_{\mathcal{F}}(\rho)\leq\overline{C}_{\mathcal{R}}(\rho)$.
Moreover, for any state $\rho$, it holds that $\overline{C}_{\mathcal{R}}(\rho)\leq\overline{C}_{l_1}(\rho)$ \cite{Piani2016}. Combining
all these arguments,
\begin{align}
\overline{C}_{l_1}(\rho)\leq C_{\mathcal{F}}(\rho)\leq\overline{C}_{\mathcal{R}}(\rho)\leq\overline{C}_{l_1}(\rho),
\end{align}
finally we have $C_{\mathcal{F}}(\rho)=\overline{C}_{\mathcal{R}}(\rho)=\overline{C}_{l_1}(\rho)$.
\end{proof}

Obviously, for a one-parameter subclass of the \textit{maximally coherent mixed states} (MCMS) \cite{Singh2015,Yao2016,Streltsov2018}
\begin{align}
\rho_m=p|\phi^+\rangle\langle\phi^+|+\frac{1-p}{d}\openone_d,
\label{MCMS}
\end{align}
this theorem also holds and thus $C_{\mathcal{F}}(\rho_m)=\overline{C}_{\mathcal{R}}(\rho_m)=\overline{C}_{l_1}(\rho_m)=[p(d-1)+1]/d$.
For more examples, one can refer to Ref. \cite{Piani2016}.

\section{Formulation of Robustness of Coherence}\label{sec3}
In Sec. \ref{sec2}, a preliminary discussion has been presented towards the relationship between $C_{\mathcal{F}}(\rho)$ and $\overline{C}_{\mathcal{R}}(\rho)$.
In this section, as will become clear later, we should take a closer look at the definition of $\overline{C}_{\mathcal{R}}(\rho)$ through an alternative expression,
which can be stated in the following theorem and the essential role of GIO is highlighted (see Fig. \ref{fig2}).
\begin{theorem}
For any $\rho \in {\mathcal{D}}(\mathbb{C}^d)$, $C_{\mathcal{R}}(\rho)$ can be cast as
\begin{align}
\overline{C}_{\mathcal{R}}(\rho)=\max_{\Lambda\in\textrm{GIO}}F[\Lambda(\rho),|\phi^+\rangle\langle\phi^+|],
\end{align}
where $F(\rho,\sigma)=\left(\textrm{tr}\sqrt{\rho^{1/2}\sigma\rho^{1/2}}\right)^2$ denotes the fidelity.
\label{T2}
\end{theorem}

\begin{figure}[htbp]
\begin{center}
\includegraphics[width=0.25\textwidth]{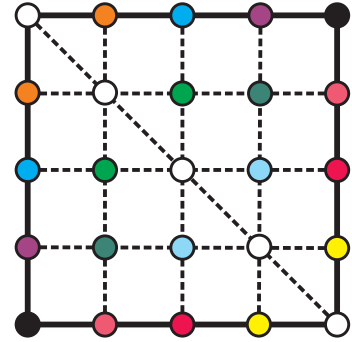}
\end{center}
\caption{(Color online) The effect of a GIO can be viewed as particular matrix \textit{sieve} which preserves the elements on the main diagonal but partially
\textit{blocks} the off-diagonal positions (e.g., white color indicates no blockage) \cite{Yao2017}.
}\label{fig2}
\end{figure}

Before proceeding, we first review the structure of GIO. In fact, in our context, the notion of GIO is equivalent to that of
the Schur channel \cite{Li1997,Paulsen2003,Watrous2018}. Suppose $\Lambda$ is a completely positive map acting on $\mathcal{H}(\mathbb{C}^d)$.
The following statements are equivalent \cite{Vicente2017}:
\begin{enumerate}
\item $\Lambda$ is a GIO (i.e., a Schur channel).
\item $\Lambda$ preserves all incoherent basis states, i.e., $\Lambda(|i\rangle\langle i|)=|i\rangle\langle i|$ for all $i$.
\item For every Kraus representation of $\Lambda(\rho)=\sum_iK_i\rho K_i^\dagger$, all Kraus operators $\{K_{i}\}$ are diagonal in the incoherent basis.
\item $\Lambda$ can be written as a Schur product form
\begin{align}
\Lambda[\rho]=\tau\circ\rho,
\label{Schur}
\end{align}
where the matrix $\tau$ is positive semidefinite such that its diagonals are all equal to 1 (i.e., $\tau_{ii}=1$)
and the Schur product of $A=[a_{ij}]$ and $B=[b_{ij}]$ is denoted by $A\circ B=[a_{ij}b_{ij}]$.
\end{enumerate}
Note that the matrix $\tau$ emerged in Eq. (\ref{Schur}) is called \textit{correlation matrix} \cite{Horn1991}. Therefore, a GIO $\Lambda$ is fully characterized
by a specified correlation matrix $\tau$, which can also be represented as a Gram matrix of a set of dynamical vectors \cite{Yao2017}.
Moreover, in Ref. \cite{Yao2017} the authors have demonstrated that the GIOs in fact constitute the core of other types of incoherent operations (see Fig. \ref{fig1}).

Now we are ready to prove Theorem \ref{T2}.

\begin{proof}
Here we continue to employ the dual form of the semidefinite program (SDP) representation of $C_{\mathcal{R}}(\rho)$ \cite{Piani2016,Chitambar2016b}
\begin{align}
1+C_{\mathcal{R}}(\rho)=\min_\tau\{\textrm{tr}(\rho\tau)\,|\, \tau\geq0, \Delta(\tau)=\openone\},
\label{dual}
\end{align}
where the constraint of $\tau$ is equivalent to requiring that $\tau$ belongs to the set of correlation matrices.
On the other hand, we notice that every GIO maps a uniform matrix (i.e., $\mathcal{J}=d|\phi^+\rangle\langle\phi^+|$) to
a unique correlation matrix:
\begin{align}
d\Lambda_{\textrm{GIO}}(|\phi^+\rangle\langle\phi^+|)=d\tau\circ|\phi^+\rangle\langle\phi^+|=\tau.
\end{align}
We also observe that the dual map of a GIO is still a GIO and can be written as $\Lambda_{\textrm{GIO}}^\dagger(\rho)=\tau^T\circ\rho$,
where the superscript $T$ denotes transpose and $\tau^T$ is also a correlation matrix. Thus, from Eq. (\ref{dual}) we have
\begin{align}
\overline{C}_{\mathcal{R}}(\rho)&=\max_{\Lambda\in\textrm{GIO}}\textrm{tr}[\rho\Lambda_{\textrm{GIO}}^\dagger(|\phi^+\rangle\langle\phi^+|)]\nonumber\\
&=\max_{\Lambda\in\textrm{GIO}}\textrm{tr}[|\phi^+\rangle\langle\phi^+|\Lambda_{\textrm{GIO}}(\rho)]\nonumber\\
&=\max_{\Lambda\in\textrm{GIO}}\langle\phi^+|\Lambda_{\textrm{GIO}}(\rho)|\phi^+\rangle.\label{GIO}
\end{align}
The proof is completed.
\end{proof}

Furthermore, $\overline{C}_{\mathcal{R}}(\rho)$ can also be written as
\begin{align}
\overline{C}_{\mathcal{R}}(\rho)=\max_{\Lambda\in\textrm{GIO}}\langle\phi|\Lambda_{\textrm{GIO}}(\rho)|\phi\rangle,
\end{align}
where $|\phi\rangle$ is an arbitrary $\mathcal{MCS}$. This equality holds since any $|\phi\rangle$ is related to $|\phi^+\rangle$
by a diagonal unitary operator $U_d$ and this unitary transformation can be absorbed into the optimization over all GIOs.
To be more precisely, let $|\nu\rangle$ be the column vector consisting of the diagonal elements of $U_d$ (e.g., $|\nu\rangle=\{e^{\textrm{i}\theta_j}\}$),
and then the action of this unitary transformation can be viewed as a particular GIO
\begin{align}
U_d\rho U_d^\dagger=|\nu\rangle\langle\nu|\circ\rho,\label{unitary}
\end{align}
where in this case $\tau=|\nu\rangle\langle\nu|$ is a \textit{rank-one} correlation matrix.
Besides, from Eqs. (\ref{GIO}) and (\ref{unitary}), it is easy to confirm
$C_{\mathcal{F}}(\rho)\leq\overline{C}_{\mathcal{R}}(\rho)$ again.

\section{Numerical simulation}\label{sec4}
Although we have clarified the relationship between $C_{\mathcal{F}}(\rho)$ and $\overline{C}_{\mathcal{R}}(\rho)$,
the \textit{tightness} of the inequality $C_{\mathcal{F}}(\rho)\leq\overline{C}_{\mathcal{R}}(\rho)$ has not been fully explored.
To gain further insight into this problem, we perform a numerical simulation aiming at verifying the gap between
$C_{\mathcal{F}}(\rho)$ and $\overline{C}_{\mathcal{R}}(\rho)$ for randomly generated states in low dimensions (see Fig. \ref{fig3}).
Our numerical approach is based on the Global Optimization Toolbox provided by {\sc matlab} \cite{https}.
For more details, we refer the readers to Appendix \ref{A2}.

For simplicity, we focus on the low-dimensional cases (e.g., $d=3,4,5,6$). Here we have randomly generated $10^4$ states for each dimension
and the numerical values of $C_{\mathcal{F}}(\rho)$ and $\overline{C}_{\mathcal{R}}(\rho)$ are plotted in Fig. \ref{fig3}. The numerical simulation shows that
although $C_{\mathcal{F}}(\rho)$ never exceed $\overline{C}_{\mathcal{R}}(\rho)$,
$C_{\mathcal{F}}(\rho)$ yields an approximate value of $\overline{C}_{\mathcal{R}}(\rho)$, at least for low-dimensional cases. For instance,
one may introduce a \textit{relative gap}, defined as
\begin{align}
g(\rho)=\frac{\overline{C}_{\mathcal{R}}(\rho)-C_{\mathcal{F}}(\rho)}{\overline{C}_{\mathcal{R}}(\rho)},
\end{align}
since $\overline{C}_{\mathcal{R}}(\rho)$ can be routinely calculated by the SDP method with high accuracy \cite{Piani2016}. In our simulations,
the largest observed value of $g$ varies from $1.59\times10^{-9}$, $0.027$, $0.033$ to $0.038$ with respect to $d=3,4,5,6$,
indicating that the relative gap tends to enlarge along with the increase of dimension.

Two observations are worth emphasizing. First, in contrast to $\overline{C}_{\mathcal{R}}(\rho)$, the evaluation of $C_{\mathcal{F}}(\rho)$
may be faced with the risk of actually obtaining a local maximum, especially for high-dimensional cases (see Appendix \ref{A2}).
Therefore, in these circumstances one may expect that the true points of some states would be slightly ``lifted'' in $C_{\mathcal{F}}$ versus $\overline{C}_{\mathcal{R}}$ diagram
(e.g., compared with the one obtained by numerical approaches) and thus the true relative gap is probably reduced compared to the observed one.
On the other hand, it is known that $C_{\mathcal{F}}(\rho)=\overline{C}_{\mathcal{R}}(\rho)$ for any qubit state.
Intriguingly, numerical results strongly evidence that $C_{\mathcal{F}}(\rho)=\overline{C}_{\mathcal{R}}(\rho)$ also for qutrit states, up to the simulation's accuracy.
In fact, we can analytically prove this conjecture.

\begin{figure*}[htbp]
\begin{center}
\includegraphics[width=16cm]{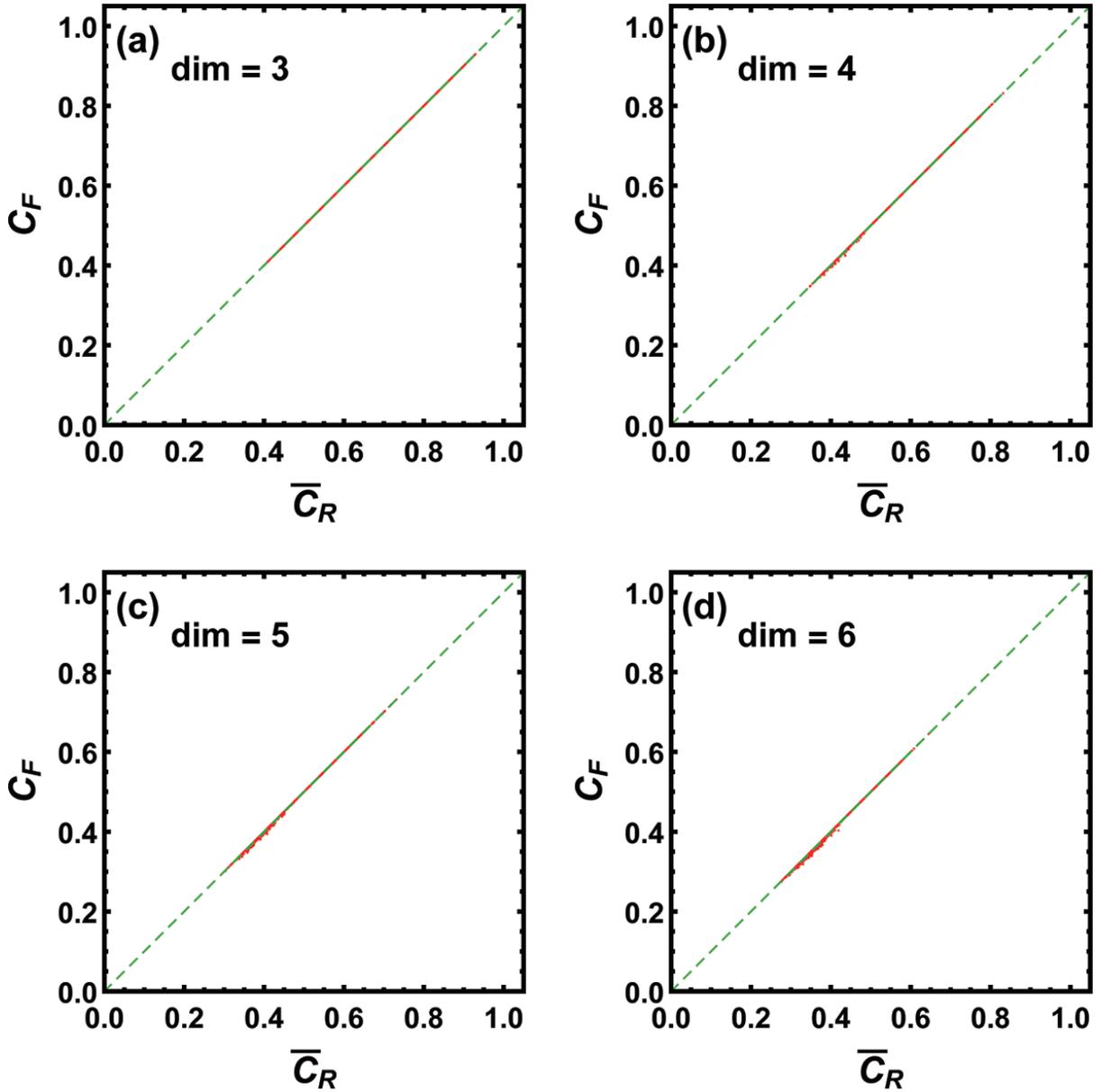}
\end{center}
\caption{(Color online) The numerical simulation (e.g., $10^4$ random states for each dimension) shows that although $C_{\mathcal{F}}(\rho)$ never exceeds $\overline{C}_{\mathcal{R}}(\rho)$,
$C_{\mathcal{F}}(\rho)$ is rather close to the value of $\overline{C}_{\mathcal{R}}(\rho)$, at least for low-dimensional cases ($d=3,4,5,6$). Interestingly, numerical results evidence that
$C_{\mathcal{F}}(\rho)=\overline{C}_{\mathcal{R}}(\rho)$ for qutrit states, up to the algorithm's accuracy.
}\label{fig3}
\end{figure*}

\begin{theorem}
For any qutrit state $\rho$, $C_{\mathcal{F}}(\rho)=\overline{C}_{\mathcal{R}}(\rho)$.
\label{T3}
\end{theorem}
\begin{proof}
We only need to prove the reverse inequality $C_{\mathcal{F}}(\rho)\geq\overline{C}_{\mathcal{R}}(\rho)$ for qutrit states.
Note that a GIO is completely characterized by the corresponding correlation matrix (e.g., a one-to-one correspondence relationship).
Moreover, the set of correlation matrices is convex and the problem of determining the extreme points of this set has been studied extensively
\cite{Christensen1979,Loewy1980,Grone1990,Li1994}. In fact, it is shown that there exists extreme correlation matrix of rank $r$
if and only if $r^2\leq d$. Thus, if $d\leq3$, the extreme correlation matrix can only be of rank-one. Here the word ``extreme''
means that any correlation matrix can be decomposed as a convex combination of these extreme points. For $d\leq3$, that is to say,
a correlation matrix $\tau$ can (only) be written as a convex pure-state decomposition
\begin{align}
\tau=\sum_kp_k|\nu_k\rangle\langle\nu_k|,
\end{align}
where $|\nu_k\rangle\langle\nu_k|$ is a rank-one correlation matrix such that every entry of $|\nu_k\rangle$ has modulus one.
Therefore, the corresponding GIO is actually a \textit{mixed unitary channel}
\begin{align}
\Lambda_{\textrm{GIO}}=\sum_kp_k|\nu_k\rangle\langle\nu_k|\circ\rho=\sum_kp_kU_k\rho U_k^\dagger,
\end{align}
where we have used Eq. (\ref{unitary}) and $U_k=\textrm{diag}\{|\nu_k\rangle\}$.
Finally, we employ the formulation of $\overline{C}_{\mathcal{R}}(\rho)$ in Theorem \ref{T2}
\begin{align}
\overline{C}_{\mathcal{R}}(\rho)&=\max_{\{U_k\}}\sum_kp_k\langle\phi^+|U_k^\dagger\rho U_k|\phi^+\rangle\\
&\leq\max_{U\in \{U_k\}}\langle\phi^+|U^\dagger\rho U|\phi^+\rangle\leq C_{\mathcal{F}}(\rho).
\end{align}
The last inequality stems from the definition of $C_{\mathcal{F}}(\rho)$.
Thus $C_{\mathcal{F}}(\rho)=\overline{C}_{\mathcal{R}}(\rho)$ for any qutrit state.
\end{proof}

The validity of Theorem \ref{T3} relies on the fact that for $d\leq3$ any GIO is a mixed unitary channel,
which was also proved in Ref. \cite{Vicente2017} by using the Choi Theorem \cite{Choi1975}.
However, here a careful scrutiny of the structure of correlation matrices actually reveals that
a GIO is a mixed unitary channel if and only if the associated correlation matrix $\tau$ can be decomposed as a
convex combination of rank-one correlation matrix. Remarkably, Theorem \ref{T3} is a natural consequence
of Theorem \ref{T2}.

\section{DISCUSSION AND CONCLUSION}\label{sec5}
In this work, we introduce the concept of quantum coherence fraction. This concept intuitively quantifies
how close a given state is to any maximally coherent state in a specified Hilbert space $\mathcal{H}(\mathbb{C}^d)$.
In addition to presenting a series of general properties of QCF, two main observations are clarified:
(i) starting from some examples illustrating the relationship between $C_{\mathcal{F}}$ and $\overline{C}_{\mathcal{R}}$,
the further research along this direction motivates us to provide an alternative formulation of $\overline{C}_{\mathcal{R}}$,
where the critical role of GIO is highlighted; (ii) by virtue of the first observation and the structure of the convex set of correlation matrices,
we have proved that for any qubit or qutrit state $\rho$, $C_{\mathcal{F}}(\rho)=\overline{C}_{\mathcal{R}}(\rho)$.

There still exist some interesting open questions. First, similar to the behavior of fully entangled fraction \cite{Badziag2000}, we
conjecture that $C_{\mathcal{F}}(\rho)$ may not satisfy the monotonicity under general incoherent operations.
However, in low-dimensional cases, it would be an exhausting work to search for an incoherent operation $\Phi$ such that
$C_{\mathcal{F}}(\Phi(\rho))\geq C_{\mathcal{F}}(\rho)$ since from Fig. \ref{fig3} it is shown that $C_{\mathcal{F}}(\rho)$
is quite close to $\overline{C}_{\mathcal{R}}(\rho)$, which is a MIO monotone \cite{Napoli2016}. Second, the simulation results
and the proof in Theorem \ref{T3} may motivate us to investigate the \textit{distance} between the GIO (i.e., the set of Schur channels)
and the set of mixtures of diagonal unitary channels, analogous to the problem raised in Ref. \cite{Yu2017}.
Moreover, in fact $C_{\mathcal{F}}(\rho)$ and $\overline{C}_{\mathcal{R}}(\rho)$ are all particular instances of
the quantity called \textit{the fidelity of coherence distillation} proposed in Ref. \cite{Regula2018a}. Therefore,
we expect that one would get a better understanding of $C_{\mathcal{F}}(\rho)$ from an operational perspective.

\begin{acknowledgments}
Y.Y. is particularly grateful for helpful discussions with G. H. Dong.
This research is supported by the Science Challenge Project (Grant No. TZ2018003-3) and
the National Natural Science Foundation of China (Grants No. 11605166 and No. 61875178).
C.P.Sun also acknowledges financial support from the NSFC (Grant No. 11534002),
the NSAF (Grants No. U1730449 and No. U1530401),
and the National Basic Research Program of China (Grants No. 2016YFA0301201 and No. 2014CB921403).
\end{acknowledgments}
\appendix

\section{Proof for properties}\label{A1}
Here we offer the proofs of the listed properties of $C_{\mathcal{F}}(\rho)$ and some further remarks are presented.

(i) (Convexity). Consider a convex decomposition of a density operator $\rho=\sum_ip_i\rho_i$, and we have
\begin{align}
C_{\mathcal{F}}(\rho)&=\max_{U\in\mathcal{U}_d}\langle\phi^+|U^\dagger\rho U|\phi^+\rangle\\
&=\langle\phi^+|U_\star^\dagger\rho U_\star|\phi^+\rangle\\
&=\sum_ip_i\langle\phi^+|U_\star^\dagger\rho_i U_\star|\phi^+\rangle\\
&\leq\sum_ip_iC_{\mathcal{F}}(\rho_i).
\end{align}
The last inequality stems from the fact that $U_\star$ may not be the optimal choice for $\rho_i$.
The convexity indicates that the mixing of states will never increase the value of $C_{\mathcal{F}}(\rho)$.

(ii) For any $\rho\in\mathcal{D}(\mathbb{C}^d)$, we suppose $C_{\mathcal{F}}(\rho)=\langle\phi_\star|\rho|\phi_\star\rangle$,
where $|\phi_\star\rangle$ is the optimal $\mathcal{MCS}$.
Note that $|\phi_\star\rangle$ can be expressed as $|\phi_\star\rangle=U_\star\mathcal{F}|0\rangle$, where $U_\star$ is a diagonal unitary operator
and $\mathcal{F}$ is the Fourier matrix with the elements $\mathcal{F}_{ij}=\omega^{ij}/\sqrt{d}$ ($\omega=e^{2\pi\textrm{i}/d}$).
Therefore, we can always construct a set of maximally coherent states including $|\phi_\star\rangle$ as an orthonormal basis in $\mathcal{H}(\mathbb{C}^d)$,
i.e., $\{|\phi_i\rangle=U_\star\mathcal{F}|i\rangle\}_{i=0}^{d-1}$. Therefore, $C_{\mathcal{F}}(\rho)$ is always larger than $1/d$ since
\begin{align}
1=\sum_i\langle\phi_i|\rho|\phi_i\rangle\leq d\langle\phi_\star|\rho|\phi_\star\rangle=dC_{\mathcal{F}}(\rho).
\end{align}

On the other side, let $\rho=\sum_i\lambda_i|\psi_i\rangle\langle\psi_i|$ be the eigen-decomposition such that $\sum_i\lambda_i=1$, $0\leq\lambda_i\leq 1$,
and $\{|\psi_i\rangle\}$ are the corresponding eigenstates. Assuming $|\phi_\star\rangle$ is the optimal $\mathcal{MCS}$, we can rewrite
\begin{align}
C_{\mathcal{F}}(\rho)=\langle\phi_\star|\rho|\phi_\star\rangle=\sum_i\lambda_i\chi_i\leq\lambda_{max},
\end{align}
where we define $\chi_i=\langle\phi_\star|\psi_i\rangle\langle\psi_i|\phi_\star\rangle$ with $\sum_i\chi_i=1$ and $\lambda_{max}$ is the largest eigenvalue.
This upper bound can be reached when $\rho$ belongs to the set of \textit{maximally coherent mixed states} (MCMS) \cite{Singh2015,Yao2016,Streltsov2018}, i.e.,
the eigenstates of $\rho$ constitute a mutually unbiased basis with respect to the incoherent basis, which implies that
in this case all eigenvectors are $\mathcal{MCS}$ and mutually orthogonal.

For instance, in a qutrit system, there exist three mutually unbiased bases (MUBs) regarding
the incoherent basis $B^{(0)}=\{|b_i^{(0)}\rangle\}=\{|0\rangle,|1\rangle,|2\rangle\}$ \cite{Bandyopadhyay2002,Durt2010}:
\begin{align}
B^{(1)}=\Big\{&|b_0^{(1)}\rangle=(1/\sqrt{3})(|0\rangle+|1\rangle+|2\rangle),\nonumber\\
&|b_1^{(1)}\rangle=(1/\sqrt{3})(|0\rangle+\omega|1\rangle+\omega^2|2\rangle),\\
&|b_2^{(1)}\rangle=(1/\sqrt{3})(|0\rangle+\omega^2|1\rangle+\omega|2\rangle)\Big\},\nonumber
\end{align}
\begin{align}
B^{(2)}=\Big\{&|b_0^{(2)}\rangle=(1/\sqrt{3})(\omega|0\rangle+|1\rangle+|2\rangle),\nonumber\\
&|b_1^{(2)}\rangle=(1/\sqrt{3})(|0\rangle+\omega|1\rangle+|2\rangle),\\
&|b_2^{(2)}\rangle=(1/\sqrt{3})(|0\rangle+|1\rangle+\omega|2\rangle)\Big\},\nonumber
\end{align}
\begin{align}
B^{(3)}=\Big\{&|b_0^{(3)}\rangle=(1/\sqrt{3})(\omega^2|0\rangle+|1\rangle+|2\rangle),\nonumber\\
&|b_1^{(3)}\rangle=(1/\sqrt{3})(|0\rangle+\omega^2|1\rangle+|2\rangle),\\
&|b_2^{(3)}\rangle=(1/\sqrt{3})(|0\rangle+|1\rangle+\omega^2|2\rangle)\Big\}.\nonumber
\end{align}
A qutrit state of the form $\rho_j=\sum_i\lambda_i|b_i^{(j)}\rangle\langle b_i^{(j)}|$ with $j\neq0$ ($\lambda_i$ is arbitrary) belongs to the MCMS class.
It is easy to verify that $C_{\mathcal{F}}(\rho_j)=\lambda_{max}$ and the pure basis state of $B^{(j)}$ corresponding to the
largest eigenvalue can be chosen as the optimal $\mathcal{MCS}$.

On the other hand, employing the method of Lagrange multipliers, it can be shown that $C_{\mathcal{F}}(\rho)=\sum_i\lambda_i\chi_i$ achieves its minimum value $1/d$ if (and only if)
$\lambda_i$ \textit{or} $\chi_i$ are all equal. For $\lambda_i=1/d$, $\rho$ is the maximally mixed state $\openone/d$, which obviously belongs to $\mathcal{I}$.
In fact, $\chi_i=1/d$ is equivalent to $|\langle\phi_\star|\psi_i\rangle|^2=1/d$.
which implies that $\{|\psi_i\rangle\}$ is mutually unbiased with respect to $|\phi_\star\rangle$. Apparently, the incoherent basis $\{|\psi_i\rangle=|i\rangle\}$
satisfies this condition and actually $|\langle i|\phi\rangle|^2=1/d$ for \textit{any} $|\phi\rangle\in\mathcal{MCS}$.

(iii) First, we recall the definitions of the coherence rank $R_c$ and coherence number $N_c$. For any $|\psi\rangle\in\mathcal{D}(\mathbb{C}^d)$,
The coherence rank $R_c(|\psi\rangle)$ refers to the number of nonzero coefficients ($c_i\neq0$) in the incoherent basis
\begin{align}
R_c(|\psi\rangle)=\min\left\{r \ \Big\vert\ |\psi\rangle=\sum_{i=1}^rc_i|\tau_i\rangle,|\tau_i\rangle\in\mathcal{I}\right\}.
\end{align}
The coherence rank is extended to the coherence number for mixed states in a similar way as
the Schmidt rank is extended to the Schmidt number \cite{Terhal2000,Sanpera2001}:
\begin{align}
N_c(\rho)=\min_{\{p_i,|\psi_i\rangle\}}\max_iR_c(|\psi_i\rangle),
\end{align}
where the optimization is over all pure-state convex decompositions of $\rho$. If $N_c(\rho)=k$, from this definition,
we know that there exists a decomposition of $\rho=\sum_ip_i|\psi_i\rangle\langle\psi_i|$ with $R_c(|\psi_i\rangle)\leq k$
for all vectors $\{|\psi_i\rangle\}$. Utilizing Eq. (\ref{pure}) and the Cauchy-Schwarz inequality, we have
$C_{\mathcal{F}}(|\psi_i\rangle)\leq k/d$. Consequently, from the convexity property of $C_{\mathcal{F}}(\rho)$, we finally arrive at
\begin{align}
C_{\mathcal{F}}(\rho)\leq\sum_ip_iC_{\mathcal{F}}(|\psi_i\rangle)\leq\frac{k}{d}=\frac{N_c(\rho)}{d}.
\end{align}

As an illustration, one can consider $\rho_m$ in Eq. (\ref{MCMS}), which is a one-parameter subset of the MCMS.
From the results in Ref. \cite{Ringbauer2018}, $\rho_m$ has coherence number $k$ if and only if
\begin{align}
\frac{k-2}{d-1}< p \leq\frac{k-1}{d-1}.
\end{align}
On the other hand, if $N_c(\rho_m)=k$, then the QCF of $\rho_m$ is bounded by
\begin{align}
\frac{k-1}{d}<C_{\mathcal{F}}(\rho_m)=\frac{(d-1)p+1}{d}\leq\frac{k}{d}=\frac{N_c(\rho_m)}{d}.
\end{align}

(iv) In Ref. \cite{Lami2019}, the author constructed a family of SIO monotones ($1\leq k\leq d$):
\begin{align}
\mu_k(\rho):=\max_{I\subseteq[d],|I|\leq k}\log\left\|\Pi_I\Delta(\rho)^{-1/2}\rho\Delta(\rho)^{-1/2}\Pi_I\right\|_\infty,
\end{align}
where $I$ is a subset of $[d]:=\{1,\ldots,d\}$, $\Pi_I=\sum_{i\in I}|i\rangle\langle i|$ and $|I|$ is the cardinality.
Note that $\mu_k(\rho)$ is a key tool for investigating the SIO distillable coherence \cite{Lami2019}.
When $k=d$, $\mu_k(\rho)$ reduces to
\begin{align}
\mu_d(\rho)&=\log\left\|\Delta(\rho)^{-1/2}\rho\Delta(\rho)^{-1/2}\right\|_\infty\\
&=\min\left\{\gamma \,|\, \rho\leq 2^\gamma \Delta(\rho)\right\}\\
&=D_{max}(\rho\|\Delta(\rho)).
\end{align}
Based on the line of thought in the proof of Lemma 8 in Ref. \cite{Lami2019}, one can show that
\begin{align}
\mu_d(\rho)&=\log\left\|\Delta(\rho)^{-1/2}\rho\Delta(\rho)^{-1/2}\right\|_\infty\nonumber\\
&\geq\log d\langle\phi|\Delta(\rho)^{1/2}\Delta(\rho)^{-1/2}\rho\Delta(\rho)^{-1/2}\Delta(\rho)^{1/2}|\phi\rangle\nonumber\\
&=\log d\langle\phi|\rho|\phi\rangle,
\end{align}
where $|\phi\rangle$ is an arbitrary $\mathcal{MCS}$ and the inequality stems from the fact that $\Delta(\rho)^{1/2}|\phi\rangle$
is always a normalized pure state. When we take the maximum value of the lower bound, this inequality is equivalent to
\begin{align}
d C_{\mathcal{F}}(\rho)\leq 2^{\mu_d(\rho)}.
\end{align}

Alternatively, this inequality can also be proved directly from the definitions of $C_{\mathcal{R}}$ and $\mu_d(\rho)$
\begin{align}
1+C_{\mathcal{R}}(\rho)=\min_\sigma\{\textrm{tr}(\sigma) \,|\, \rho\leq\sigma, \sigma=\Delta(\sigma)\}\leq2^{\mu_d(\rho)},
\end{align}
since $\Delta(\rho)\in\{\sigma:\sigma=\Delta(\sigma)\geq0\}$. Moreover, in Theorem \ref{T1},
we have proved that $dC_{\mathcal{F}}(\rho)\leq1+C_{R}(\rho)$ and thus $d C_{\mathcal{F}}(\rho)\leq 2^{\mu_d(\rho)}$.

\section{Numerical simulation}\label{A2}
Setting $U_d=\textrm{diag}\{e^{\textrm{i}\theta_j}\}$ and $\theta_{ij}=\theta_i-\theta_j$, Eq. (\ref{QCF}) can be recast as
\begin{align}
C_{\mathcal{F}}(\rho)&=\max_{U\in\mathcal{U}_d}\sum_{i,j}\left[U^\dagger\rho U\right]_{ij}
=\max_{\{\theta_i\}}\sum_{i,j}\rho_{ij}e^{-\textrm{i}\theta_{ij}}\nonumber\\
&=1+\max_{\{\theta_i\}}2\sum_{i<j}\left[\Re(\rho_{ij})\cos(\theta_{ij})+\Im(\rho_{ij})\sin(\theta_{ij})\right]\nonumber\\
&:=1+2\max_{\{\theta_i\}}f(\theta_0,\theta_1,\ldots,\theta_{d-1}),
\label{numerical}
\end{align}
where $\Re(z)$ and $\Im(z)$ denote the real and imaginary parts of a complex number $z$ respectively.
Since the overall phase is irrelevant, one can set $\theta_0=0$ and thus the optimization in Eq. (\ref{numerical}) is actually over $d-1$ phases.

Furthermore, from Eq. (\ref{numerical}) it is demonstrated that the objective function $f(\theta_1,\ldots,\theta_{d-1})$ to be optimized
is a linear combination of $\sin(\theta_{ij})$ and $\cos(\theta_{ij})$ functions, which are all periodic functions. Therefore, it is easy to verify
that the objective function is in fact a \textit{multi-peak} function and there exists a considerable amount of \textit{local} maximums.
To obtain a global maximum, we need a robust and reliable global optimization algorithm for maximizing the function $f(\theta_1,\ldots,\theta_{d-1})$,
where we choose to employ the the Global Optimization Toolbox provided by {\sc matlab} \cite{https}. In low dimensions,
all the runs are rather successful.


\end{document}